\documentclass[12pt]{article}

\usepackage[utf8]{inputenc}
\usepackage[english]{babel}
\usepackage[margin=1in]{geometry}
\usepackage[dvipsnames,table]{xcolor}
\usepackage{amssymb,amsmath,amsthm,setspace,graphicx}
\usepackage{booktabs}
\usepackage{multirow}
\usepackage{rotating}
\usepackage{array}
\usepackage{siunitx}
\usepackage[justification=centering]{caption}
\usepackage[authoryear]{natbib}
\usepackage{url}
\usepackage{float}
\usepackage[section]{placeins}
\usepackage{mathtools}
\usepackage{enumitem}
\usepackage[normalem]{ulem}
\usepackage{algorithm}
\usepackage{algpseudocode}
\usepackage[pdfborder={0 0 0},final=true,colorlinks=true,breaklinks=true,linkcolor=red,citecolor=blue]{hyperref}
\usepackage[capitalize,nameinlink,noabbrev]{cleveref}
\usepackage{relsize}
\usepackage{setspace}

\newcommand{\tightdisplayspacing}{%
	\setlength{\abovedisplayskip}{0.35\baselineskip}%
	\setlength{\belowdisplayskip}{0.35\baselineskip}%
	\setlength{\abovedisplayshortskip}{0.20\baselineskip}%
	\setlength{\belowdisplayshortskip}{0.20\baselineskip}%
	\setlength{\jot}{2pt}%
}
\AtBeginDocument{\tightdisplayspacing}

\renewenvironment{abstract}{%
	\begin{center}%
		\textbf{\abstractname}%
		\vspace{0.5\baselineskip}%
	\end{center}%
	\begin{singlespace}%
		\noindent\ignorespaces%
	}{%
	\end{singlespace}%
}

\numberwithin{equation}{section}

\DeclareMathOperator*{\argmax}{arg\,max}
\DeclareUnicodeCharacter{2212}{\ensuremath{-}}

\theoremstyle{plain}
\newtheorem{thm}{Theorem}[section]
\crefname{thm}{theorem}{theorems}
\Crefname{thm}{Theorem}{Theorems}

\crefname{cor}{corollary}{corollaries}
\Crefname{cor}{Corollary}{Corollaries}
\newtheorem{prop}{Proposition}
\Crefname{prop}{Proposition}{Propositions}
\newtheorem{lem}{Lemma}
\crefname{lem}{lemma}{lemmas}
\Crefname{lem}{Lemma}{Lemmas}
\theoremstyle{definition}
\newtheorem{axiom}{Axiom}[section]
\crefname{axiom}{axiom}{axioms}
\Crefname{axiom}{Axiom}{Axioms}

\crefname{defn}{definition}{definitions}
\Crefname{defn}{Definition}{Definitions}
\theoremstyle{remark}
\newtheorem{rem}{Remark}[section]
\crefname{rem}{remark}{remarks}
\Crefname{rem}{Remark}{Remarks}

\definecolor{lightgrey}{rgb}{0.83,0.83,0.83}
\graphicspath{{./}{../}}

\begin{document}
	
	\pagenumbering{gobble}
	
	\begin{singlespace}
		
		\title{\vspace{-4ex}\textbf{Reference Dependence and the Structure of the WTA/WTP Gap}}
		\author{G. Charles-Cadogan
			\thanks{School of Accounting and Finance, University of Leicester; Tel: +44 (0116) 229 7385; e-mail: \href{mailto:gocadog@gmail.com}{gocadog@gmail.com}
			\\~\\
			The author has no conflicts of interest--financial or otherwise associated with this paper.\\~\\
			Declaration of AI use. During the preparation of this manuscript, the author used generative AI tools to assist with language editing and stylistic refinement. The author reviewed and edited all content and takes full responsibility for the final manuscript.
		} 
	} 
		
		\date{\today}
		\maketitle
		
		\begin{abstract}
			\noindent This paper studies the willingness-to-accept/willingness-to-pay (WTA-WTP) gap under objective probabilities. Preferences over finite lotteries satisfy completeness, transitivity, continuity, weak independence, reference partition, and range dependence. Weak independence requires von Neumann-Morgenstern independence only for mixtures that preserve the reference point and do not move outcomes across the induced gain-loss partition. The representation, weak rank-dependent utility (WRDU), evaluates gains and losses by separate subutilities anchored at the reference point and recombines them through a range-dependent Lagrangian penalty coefficient $\rho$ on the loss-side component. The reciprocal index $\lambda=1/\rho$ reports the WTA-WTP loss-aversion convention. The main result characterizes a normalized admissible transaction class in which the WTA-WTP gap follows from the asymmetric buying and selling indifference equations. In this class, $\lambda>1$ suppresses WTP and elevates WTA, while $\rho>1$ corresponds to gain seeking or loss attenuation. A fixed reciprocal loss-aversion index has no internal mechanism that makes the wedge converge to zero as transaction scale changes; attenuation requires a transaction path on which $\rho$ and $\lambda$ converge to their common neutral value one. The analysis gives a decision-theoretic account of the endowment-effect wedge based on weakened independence, reference anchoring, semi-affine subutility normalization, and range-dependent penalization. The result is distinct from the Rabin calibration implication and does not rely on constant-relative-risk-aversion utility or probability weighting.
			
			\vspace{1ex}
			
			\vspace{1ex}
		\end{abstract}
		
	\end{singlespace}
	
	\newpage
	\pagenumbering{roman}
	\thispagestyle{empty}
	\singlespace
	\tableofcontents
	\onehalfspacing
	
	\newpage
	\pagenumbering{arabic}
	\hypersetup{pageanchor=true}
	\tightdisplayspacing
	
	\section{Introduction}
	
	The endowment effect--the systematic disparity between willingness to accept (WTA) and willingness to pay (WTP)--is one of the most robust empirical regularities in behavioral economics. Beginning with the seminal experiments of \citet{KahnemanKnetschThaler1990}, researchers have consistently found that individuals demand more to give up a good than they are willing to pay to acquire it. Meta-analyses \citep{HorowitzMcConnell2002,TuncelHammitt2014} report large WTA-WTP ratios, especially for non-market goods. The gap remains theoretically important because standard equilibrium reasoning assigns the same value to the same object when income, ownership, and transaction costs are properly controlled. See also \citet{Frederick2012,Depoorter2025}.
	
	The standard explanation in Cumulative Prospect Theory (CPT) \citep{TverKahn1992} attributes this gap to an exogenous loss-aversion parameter $\lambda > 1$: losses loom larger than gains. However, this explanation faces two difficulties. First, it provides no structural reason for why the gap should vary with transaction scale. Second, recent evidence \citep{ChapmanDeanOrtolevaSnowbergCamerer2024} shows that the endowment effect and loss aversion for risk are decoupled.
	
	This paper provides a decision-theoretic account of the WTA-WTP gap using weak rank-dependent utility (WRDU). The model is abstract and does not depend on constant-relative-risk-aversion (CRRA) utility or any other parametric form. The result is stated for a normalized admissible transaction class. Within that class, the gap arises from the structural asymmetry of the indifference conditions for buying and selling. Equation \eqref{eq:WRDU_Representation} uses the Lagrangian penalty coefficient $\rho$ on the loss-side component. The reciprocal index $\lambda=1/\rho$ is then used to report WTA-WTP comparative statics in the standard loss-aversion convention: loss aversion corresponds to $\lambda>1$, while $\lambda<1$ corresponds to gain seeking or loss attenuation. Under comparable gain and loss subutilities, $\lambda>1$ lowers WTP and raises WTA.
	
	Crucially, the Lagrangian coefficient $\rho$ is not fixed. It is selected by the reference-point problem and varies with the transaction range. Alessandro Cerboni characterized the range as the ``emotional distance". Let $m$ denote transaction scale, let $\rho(m)$ denote the Lagrangian penalty coefficient evaluated along that transaction path, and let $\lambda(m)=1/\rho(m)$ denote the reciprocal loss-aversion index used in the WTA-WTP equations. For the WTA-WTP application, attenuation means that $\lambda(m)$ and $\rho(m)$ both move toward the neutral value one as the transaction becomes familiar, large, or market-like. This is distinct from the Rabin calibration result, where the large-stakes implication concerns favorable risky gambles rather than the buy-sell wedge studied here.
	
	The paper is organized as follows. \cref{sec:LitRev} briefly reviews the relevant literature on WTA-WTP disparities. \cref{sec:WTA_WRDU_Axioms} states the axiomatic foundations of WRDU. \cref{sec:WRDUmodel_specification} presents the WRDU model for WTA and WTP. \cref{sec:AbstractResolveWTA_WTP_gap} derives the abstract structural resolution of the WTA-WTP gap. \cref{sec:Illustrative_example} provides a concrete illustrative example. \cref{sec:Conclusion} concludes. Proofs and supplementary material are collected in the Appendix.
	
	\section{Literature Review}\label{sec:LitRev}
	
	The WTA-WTP gap has been extensively documented. \citet{KahnemanKnetschThaler1990} found median WTA of \$7.12 versus median WTP of \$2.87 for coffee mugs. \citet{HorowitzMcConnell2002} reported average WTA-WTP ratios around 7:1.
	
	Several explanations have been proposed. The standard account in Prospect Theory attributes the gap to loss aversion: the pain of losing a good exceeds the pleasure of gaining it \citep{KahnTver1979,TverKahn1992}. However, recent evidence challenges this view. \citet{ChapmanDeanOrtolevaSnowbergCamerer2024} found no correlation between the endowment effect and loss aversion for risk in large representative samples, suggesting the two phenomena are decoupled.
	
	Alternative explanations include substitution effects \citep{Hanemann1991}, market experience \citep{List2003}, and information effects \citep{WeaverFrederick2012}. More recent work by \citet{CerreiaVioglioDillenbergerOrtoleva2024} develops Cautious Utility, where the gap arises from uncertainty about trade-offs and caution rather than loss aversion.
	
	The WRDU approach differs from these alternatives by providing a structural mechanism rooted in utility curvature asymmetry and range-dependent penalization. On the normalized admissible transaction class, it explains how the gap can arise and how it can attenuate with transaction scale, without introducing a fixed Lagrangian penalty coefficient or a fixed reciprocal loss-aversion index. This places WRDU within the recent literature on reference-dependent preferences and behavioral welfare analysis \citep{EricksonFuster2014}.
	
	\section{Axiomatic Foundations of WRDU}\label{sec:WTA_WRDU_Axioms}
	
	This section makes the WTA-WTP application self-contained by stating the preference restrictions used by WRDU. Let $\mathcal L(X)$ denote the set of finite lotteries over a compact outcome interval $X$, and let $\succeq$ be a preference relation on $\mathcal L(X)$. The axioms preserve objective probabilities, weaken independence only where reference dependence requires it, and impose an endogenous gain-loss partition.
	
	\subsection{Preference Axioms}\label{subsec:PrefAxioms}
	
	\begin{axiom}[Completeness]\label{axiom:WTA_Completeness}
		For any $L,M\in\mathcal L(X)$, either $L\succeq M$, $M\succeq L$, or both.
	\end{axiom}
	
	\begin{axiom}[Transitivity]\label{axiom:WTA_Transitivity}
		For any $L,M,N\in\mathcal L(X)$, if $L\succeq M$ and $M\succeq N$, then $L\succeq N$.
	\end{axiom}
	
	\begin{axiom}[Continuity]\label{axiom:WTA_Continuity}
		For any $L,M,N\in\mathcal L(X)$, the sets
		\begin{align}
		&\{\alpha\in[0,1]:\alpha L+(1-\alpha)M\succeq N\},\\
		&\{\alpha\in[0,1]:N\succeq \alpha L+(1-\alpha)M\}
		\end{align}
		are closed.
	\end{axiom}
	
	\begin{axiom}[Weak Independence]\label{axiom:WTA_WeakIndependence}
		The preference relation satisfies the following two restrictions.
		
		\smallskip
		\noindent\textup{(i) Partition-preserving mixture independence.}
		Let $L,M,N\in\mathcal L(X)$ and $\alpha\in[0,1]$. If $L\succeq M$, and if mixing both $L$ and $M$ with $N$ leaves the reference point unchanged and does not move any outcome across the gain-loss partition, then
		\begin{equation}
		\alpha L+(1-\alpha)N\succeq \alpha M+(1-\alpha)N.
		\end{equation}
		
		\smallskip
		\noindent\textup{(ii) Gain-loss lottery-pair equivalence.}
		Fix a lottery $L=\{(x_i,p_i)\}_{i=1}^n$ with support $X$ and reference point $x_r$. Let
		\begin{align}
			P_g(L)&=\sum_{x_i\in C_g(X)}p_i,\\
			P_\ell(L)&=\sum_{x_i\in C_\ell(X)}p_i,\\
			P_r(L)&=1-P_g(L)-P_\ell(L).
		\end{align}
		Let $C_g(X)$ denote the gain-side component induced by the support, $C_\ell(X)$ the loss-side component, and let $\widehat C_r(X)$ denote the reference cell, which reduces to $\{x_r\}$ in the WTA-WTP transaction application. The compound lottery that keeps the gain, loss, and reference components separated is equivalent to the simple reference lottery only through the penalized gain-loss balance:
		\begin{equation}
			\left\{(P_r(L),\widehat C_r(X))\right\}
			\sim
			\left\{(P_g(L),C_g(X));\,
			(P_\ell(L),-\rho(x_r;y,z)C_\ell(X))\right\}.
			\label{eq:WTA_GainLossPairEquivalence}
		\end{equation}
		Equivalently, the reference component aggregates with the gain-side and loss-side components according to
		\begin{equation}
			P_r(L)\widehat C_r(X)
			\sim
			P_g(L)C_g(X)
			-\rho(x_r;y,z)P_\ell(L)C_\ell(X).
		\end{equation}
		Thus independence is imposed within a fixed reference partition, while cross-partition aggregation is mediated by the range-dependent Lagrangian penalty coefficient $\rho(x_r;y,z)$. This is the WTA-WTP analogue of the generalized lottery-pair equivalence in \citet{CharlesCadogan2016}. The associated reciprocal index is $\lambda(x_r;y,z)=1/\rho(x_r;y,z)$.
	\end{axiom}
	
	\begin{axiom}[Reference Partition]\label{axiom:WTA_ReferencePartition}
		For each choice problem with support $X$, there exists a reference point $x_r\in X$ such that
		\begin{equation}
		X=C_\ell(X)\cup\{x_r\}\cup C_g(X),
		\end{equation}
		where $C_\ell(X)=\{x\in X:x<x_r\}$ and $C_g(X)=\{x\in X:x>x_r\}$. The reference point is selected as a maximizer of the Lagrangian penalty functional
		\begin{equation}
		x_r\in\argmax_{r\in X}\left[v_g(z(r))-\rho(r)v_\ell(y(r))\right],
		\end{equation}
		where $y(r)$ and $z(r)$ are the relevant loss and gain magnitudes induced by the candidate reference point $r$, and $\rho(r)=1/\lambda(r;y(r),z(r))$ is the Lagrangian penalty coefficient. If the first-order condition has several roots, the root with the largest Lagrangian value is selected.
	\end{axiom}
	
	\begin{axiom}[Range Dependence]\label{axiom:WTA_RangeDependence}
		For relevant loss and gain magnitudes $y,z>0$, the virtual loss-aversion index $\lambda(x_r;y,z)>0$ satisfies
		\begin{align}
		\frac{\partial\lambda}{\partial z}&>0,\\
		\frac{\partial\lambda}{\partial y}&<0,
		\end{align}
		and varies with the transaction range rather than remaining fixed across choice problems. The neutral benchmark is $\lambda(x_r;y,z)=1$, equivalently $\rho(x_r;y,z)=1/\lambda(x_r;y,z)=1$.
	\end{axiom}
	
	\subsection{Representation}\label{subsec:RepresentationTheorem}
	
	\begin{thm}[WRDU Representation]\label{thm:WRDURepresentation}
		Let $\succeq$ be a preference relation on $\mathcal L(X)$ satisfying
		Completeness, Transitivity, Continuity, Weak Independence,
		Reference Partition, and Range Dependence.
		
		Suppose, in addition, that the region-wise representatives are selected
		from the admissible strictly concave $C^2$ WRDU transaction class.
		Then on mixtures that preserve the reference partition,
		there exist strictly increasing, strictly concave functions
		$v_g,v_\ell:\mathbb R_+\to\mathbb R_+$ with
		$v_g(0)=v_\ell(0)=0$
		and a positive range-dependent Lagrangian penalty $\rho(x_r;y,z)$, with $\lambda(x_r;y,z)=1/\rho(x_r;y,z)$,
		such that preferences admit the representation
		\begin{equation}
			V^\rho_{WRDU}(L)=
			\sum_{x\in C_g(X)}p_xv_g(x-x_r)
			-\rho(x_r;y,z)
			\sum_{x\in C_\ell(X)}p_xv_\ell(x_r-x).
		\end{equation}
		
		Moreover:
		
		(i) $(v_g,v_\ell,\rho,\lambda)$ are pinned down once the reference-point normalization and a common scale are fixed, with $\lambda=1/\rho$.
		
		(ii) Independent additive transformations of $v_g$ and $v_\ell$ are not admissible.
		
		(iii) If $\rho$, equivalently $\lambda=1/\rho$, is constant across transaction ranges,
		preferences reduce to a fixed-Lagrangian-penalty reference-dependent model.
	\end{thm}
	
	The proof is given in \Cref{app:Proof_WRDURepresentation}.
	
	\begin{rem}[Why Standard von Neumann-Morgenstern Affine Freedom Does Not Apply Separately]\label{rem:WTA_SemiAffine}
		The original utility representation is affine in the usual von Neumann-Morgenstern sense. After the reference split, however, $v_g$ and $v_\ell$ are anchored at $x_r$. Positive scaling remains admissible, but independent translations would move the anchors and change the gain-loss decomposition. The subutilities are therefore semi-affine representatives; this is the sense in which anchoring matters for WTA-WTP comparisons \citep{TverskyKahneman1974}.
	\end{rem}
	
	\begin{rem}[Scope of the Axioms]\label{rem:WTA_AxiomScope}
		The WTA-WTP result below is not a global claim about all models of reference dependence. It is a theorem about the normalized WRDU transaction class satisfying the stated axioms. The stronger endpoint restrictions used in the Allais and Rabin applications are not needed for the WTA-WTP wedge; here the relevant condition is whether $\lambda=1/\rho$ lies above, below, or near its neutral value. The coefficient $\rho$ is the Lagrangian penalty in \eqref{eq:WRDU_Representation}; it is also useful as a diagnostic for gain-seeking or excessive-risk-seeking paths when $\lambda<1$.
	\end{rem}
	
	\begin{rem}
		Uniqueness up to common positive scaling reflects the semi-affine structure induced by reference anchoring. Because $v_g(0)=v_\ell(0)=0$, independent translations would alter the gain-loss partition and therefore change the preference ordering. This distinguishes the WRDU class from standard von Neumann-Morgenstern representations.
	\end{rem}
	
	\section{The WRDU Model}\label{sec:WRDUmodel_specification}
	
	This section applies the axiomatic representation in \Cref{thm:WRDURepresentation} to WTA and WTP.
	
	\subsection{Basic Setup}
	
	Let $X$ be a compact interval of prizes. A finite lottery is denoted by $L=\{(x_i,p_i)\}_{i=1}^n$, where $x_i\in X$, $p_i\ge0$, and $\sum_i p_i=1$. Equivalently, write $p_x$ for the probability assigned to outcome $x$ in the support of $L$. The WRDU representation splits the payoff support at an endogenous reference point $x_r$ into gain and loss regions:
	
	\begin{equation}
	X = C_\ell(X) \cup \{x_r\} \cup C_g(X),
	\end{equation}
	
	where $C_\ell(X)=\{x\in X:x<x_r\}$ contains outcomes below the reference point and $C_g(X)=\{x\in X:x>x_r\}$ contains outcomes above the reference point. For any $x\in C_\ell(X)$, the loss magnitude is $x_r-x$; for any $x\in C_g(X)$, the gain magnitude is $x-x_r$.
	
	\subsection{The Penalized Utility Representation}\label{subsec:PenalizedUtility}
	
	Define the Lagrangian penalty coefficient by
	\begin{equation*}
		\rho(x_r;y,z)=\frac{1}{\lambda(x_r;y,z)}.
	\end{equation*}
	The WRDU value of a lottery $L$ is:
	
	\begin{equation}
		V^\rho_{WRDU}(L)
		=
		\sum_{x \in C_g(X)} p_x v_g(x-x_r)
		-
		\rho(x_r;y,z)
		\sum_{x \in C_\ell(X)} p_x v_\ell(x_r-x),
		\label{eq:WRDU_Representation}
	\end{equation}
	
	where:
	
	\begin{itemize}
		\item $v_g:\mathbb{R}_+ \to \mathbb{R}_+$ is the gain subutility, strictly increasing and strictly concave ($v_g'' < 0$), with $v_g(0) = 0$;
		\item $v_\ell:\mathbb{R}_+ \to \mathbb{R}_+$ is the loss subutility, defined on loss magnitudes, strictly increasing and strictly concave ($v_\ell'' < 0$), with $v_\ell(0) = 0$;
		\item $y$ denotes the relevant loss magnitude, $z$ denotes the relevant gain magnitude, $\rho(x_r;y,z)$ is the Lagrangian penalty on the loss-side component, and $\lambda(x_r;y,z)=1/\rho(x_r;y,z)$ is the corresponding range-dependent loss-aversion index under the reciprocal normalization.
	\end{itemize}
	
	The coefficient $\rho$ is the penalty coefficient in the Lagrangian representation. The reciprocal $\lambda=1/\rho$ is the associated loss-aversion index used when the model is translated into the standard WTA-WTP comparison below. The neutral benchmark is $\lambda=\rho=1$.
	
	\subsection{The Lagrangian Penalty and Reference-Point Selection}\label{subsec:LagrangianPenalty}
	
	The Lagrangian penalty coefficient has a direct optimization interpretation. For a candidate reference point $r\in X$, let $z(r)$ denote the relevant gain-side magnitude and let $y(r)$ denote the relevant loss-side magnitude induced by the partition at $r$. The reference-point problem can be written as the Lagrangian relaxation
	\begin{equation}
		\mathcal L(r;\rho)=v_g(z(r))-\rho(r)v_\ell(y(r)),
		\qquad \rho(r)\ge0.
		\label{eq:WTA_LagrangianPenalty}
	\end{equation}
	The multiplier $\rho(r)$ is the shadow-price penalty on the loss-side component. Writing
	\begin{equation}
		\rho(r)=\frac{1}{\lambda(r;y(r),z(r))}
		\label{eq:WTA_RhoLambdaDuality}
	\end{equation}
	gives the dual form
	\begin{equation}
		\mathcal L(r;1/\lambda)
		=
		v_g(z(r))-\frac{1}{\lambda(r;y(r),z(r))}v_\ell(y(r)).
		\label{eq:WTA_LagrangianLambda}
	\end{equation}
	This Lagrangian form is used to select the reference point when the first-order condition admits more than one root. The first-order condition identifies candidate roots; the reference point is the root that maximizes \eqref{eq:WTA_LagrangianPenalty}. Thus optimality, not merely stationarity, selects $x_r$:
	\begin{equation}
		x_r\in\argmax_{r\in X}\mathcal L(r;\rho).
		\label{eq:WTA_LagrangianRootSelection}
	\end{equation}
	After the reference point and path-dependent tradeoff are selected, the WTA-WTP indifference equations below are reported in the standard Wakker loss-aversion normalization, where $\lambda>1$ means loss aversion and $\rho=1/\lambda$ is the corresponding Lagrangian penalty coefficient.
	
	\subsection{The Endogenous Loss-Aversion Index}\label{subsec:EndogenLAindex}

	The reference point $x_r$ is selected as the maximizer of the Lagrangian penalty functional in \eqref{eq:WTA_LagrangianPenalty}. For a candidate reference point $r\in X$, let $y(r)$ denote the relevant loss magnitude below $r$ and let $z(r)$ denote the relevant gain magnitude above $r$. Then
	\begin{equation}
		x_r
		\in
		\argmax_{r \in X}
		\left[v_g(z(r))-\rho(r)v_\ell(y(r))\right],
		\qquad
		\rho(r)=\frac{1}{\lambda(r;y(r),z(r))}.
		\label{eq:RefPoint_Maximizer}
	\end{equation}
	The WTA-WTP valuation equations are subsequently written in the standard loss-aversion normalization, where $\lambda=1/\rho$.
	
	At an interior reference point, the local loss-aversion index is measured by the marginal loss-gain ratio:
	
	\begin{equation}
		\lambda(x_r; y, z) = \frac{v_\ell'(y)}{v_g'(z)}.
		\label{eq:lambda_FOC}
	\end{equation}
	If several candidate reference points satisfy \eqref{eq:lambda_FOC}, WRDU does not treat them as equally admissible. The selected reference point is the candidate that maximizes the Lagrangian value in \eqref{eq:WTA_LagrangianRootSelection}. This is the role of the penalty function in the optimality problem: it ranks stationary roots by their penalized gain-loss value.
	
	\subsection{Range Dependence and Boundary Conditions}
	
	The admissible WRDU class satisfies the following range-dependence properties:
	
	\begin{enumerate}
		\item \emph{Monotonicity in the gain outcome:} $\frac{\partial \lambda}{\partial z} > 0$;
		\item \emph{Monotonicity in the loss outcome:} $\frac{\partial \lambda}{\partial y} < 0$;
		\item \emph{Range responsiveness:} $\lambda(x_r;y,z)$ varies with the transaction range rather than remaining fixed.
		\item \emph{Neutral benchmark:} $\lambda(x_r;y,z)=1$, equivalently $\rho(x_r;y,z)=1/\lambda(x_r;y,z)=1$, corresponds to no WTA-WTP distortion from the WRDU transaction normalization.
	\end{enumerate}
	
	The stronger endpoint restrictions used to block the Rabin calibration implication are discussed in the broader WRDU paper. For WTA-WTP, the relevant object is whether $\lambda=1/\rho$ lies above, below, or near its neutral value. Under the standard WTA-WTP loss-aversion convention, $\lambda>1$ means that losses loom larger than gains in the buy-sell indifference equations. In the Lagrangian representation \eqref{eq:WRDU_Representation}, this same case corresponds to $\rho<1$. Conversely, $\rho>1$ corresponds to $\lambda<1$, which indicates gain seeking or loss attenuation.
	
	\section{Abstract Structural Resolution of the WTA-WTP Gap}\label{sec:AbstractResolveWTA_WTP_gap}
	
	This section provides a parameter-free structural explanation of the WTA-WTP gap, relying only on the axioms and the structural form of WRDU.
	
	\subsection{The Indifference Equations}\label{subsec:IndiffEquations_WTA_WTP}
	
	Consider a transaction involving $m$ units of a good. The reference point is the status quo $(0,0)$. In this section, $\rho(m)$ denotes the Lagrangian penalty coefficient selected by the WRDU reference-point problem, and
	\begin{equation}
		\lambda(m)=\frac{1}{\rho(m)}
	\end{equation}
	denotes the reciprocal loss-aversion index used to report the WTA-WTP indifference equations. Thus $\rho$ remains the coefficient in the WRDU Lagrangian representation, while $\lambda$ is a reciprocal reporting convention for the buy-sell comparison. After the reference point and the Lagrangian penalty have been selected, the transaction comparison is reported in the equivalent buy-sell normalization
	\begin{equation}
		U^T(g,\ell;\lambda)=v_g(g)-\lambda v_\ell(\ell),
		\qquad \lambda=\frac{1}{\rho}.
		\label{eq:TransactionReportingNormalization}
	\end{equation}
	Equation \eqref{eq:TransactionReportingNormalization} is not a second representation theorem; it is the WTA-WTP reporting normalization used to express buying and selling indifference in the standard loss-aversion convention.
	
	\paragraph{Willingness to Pay (WTP).} To buy $m$ units for price $p$:
	
	\begin{itemize}
		\item Gain: receives $m$ units $\to$ utility $v_g(m)$;
		\item Loss: pays $p$ dollars $\to$ transaction-normalized disutility $-\lambda v_\ell(p)$, with $\lambda=1/\rho$.
	\end{itemize}
	
	Indifference determines WTP:
	
	\begin{equation}
		v_\ell(\text{WTP}) = \frac{1}{\lambda} \, v_g(m).
		\label{eq:WTP_Equation}
	\end{equation}
	
	\paragraph{Willingness to Accept (WTA).} To sell $m$ units for price $r$:
	
	\begin{itemize}
		\item Gain: receives $r$ dollars $\to$ utility $v_g(r)$;
		\item Loss: gives up $m$ units $\to$ transaction-normalized disutility $-\lambda v_\ell(m)$, with $\lambda=1/\rho$.
	\end{itemize}
	
	Indifference determines WTA:
	
	\begin{equation}
		v_g(\text{WTA}) = \lambda \, v_\ell(m).
		\label{eq:WTA_Equation}
	\end{equation}
	
	The foregoing indifference equations motivate the following impossibility result for a constant reciprocal loss-aversion index.
	
	\begin{prop}[Impossibility under a Fixed Reciprocal Loss-Aversion Index]\label{prop:FixedLambdaImpossibility}
		Suppose $v_g,v_\ell$ are strictly increasing, strictly concave, and twice continuously differentiable with $v_g(0)=v_\ell(0)=0$. Suppose the WTA-WTP indifference equations are reported with a fixed reciprocal loss-aversion index $\bar\lambda>0$, equivalently with a fixed Lagrangian penalty coefficient $\bar\rho=1/\bar\lambda$, so that for every transaction scale $m>0$,
		\begin{equation}
			v_\ell(\mathrm{WTP}(m))=\frac{1}{\bar\lambda}v_g(m)=\bar\rho v_g(m),
			\qquad
			v_g(\mathrm{WTA}(m))=\bar\lambda v_\ell(m)=\frac{1}{\bar\rho}v_\ell(m).
			\label{eq:FixedLambdaIndiff}
		\end{equation}
		
		Assume further that transactions are normalized so that gain and loss subutilities are comparable across scales and that the neutral benchmark ($\bar\lambda=1$) yields $\mathrm{WTA}(m)=\mathrm{WTP}(m)$ for all $m>0$.
		
		If $\bar\lambda\neq1$, then attenuation of the WTA-WTP wedge is not
		implied by the fixed-index indifference equations. In particular, any convergence of
		\begin{equation}
			\mathrm{WTA}(m)-\mathrm{WTP}(m)
		\end{equation}
		to zero must come from additional restrictions on $v_g$ and $v_\ell$, not from
		the fixed reciprocal loss-aversion index itself.
	\end{prop}
	
	The proof is given in \Cref{app:Proof_FixedLambdaImpossibility}.

	\subsection{The Structural Origin of the Gap}\label{subsec:Structural_origin_of_gap}
	
	Compare \eqref{eq:WTP_Equation} and \eqref{eq:WTA_Equation}. Since $\rho=1/\lambda$, the WTA-WTP-normalized equations can also be written in terms of the Lagrangian penalty coefficient:
	
	\begin{itemize}
		\item WTP solves $v_\ell(\text{WTP})=v_g(m)/\lambda=\rho v_g(m)$;
		\item WTA solves $v_g(\text{WTA})=\lambda v_\ell(m)=v_\ell(m)/\rho$.
	\end{itemize}
	
	If $\lambda>1$, equivalently $\rho<1$, and if the gain and loss subutilities are normalized on a common transaction scale, then:
	
	\begin{itemize}
		\item From \eqref{eq:WTP_Equation}: $v_\ell(\text{WTP}) < v_g(m)$; since $v_\ell$ is increasing, WTP is relatively small;
		\item From \eqref{eq:WTA_Equation}: $v_g(\text{WTA}) > v_\ell(m)$; since $v_g$ is increasing, WTA is relatively large.
	\end{itemize}
	
	Thus, on the normalized admissible region where the two inverse-utility comparisons are ordered, $\lambda>1$ implies:
	
	\begin{equation}
		\text{WTA} > \text{WTP}.
		\label{eq:Gap}
	\end{equation}
	
	This is the endowment effect. No parametric form is needed, but the conclusion is conditional on the standard normalization that makes the gain and loss subutilities comparable across the good and money dimensions.
	
	\begin{rem}[Interpretation of $\lambda > 1$]
		When $\lambda > 1$, losses are magnified relative to gains in the WTA-WTP indifference equations, which is the standard behavioral interpretation of loss aversion. In the Lagrangian representation \eqref{eq:WRDU_Representation}, the corresponding penalty coefficient is $\rho=1/\lambda<1$. Conversely, $\lambda<1$ corresponds to gain seeking or loss attenuation, not loss aversion. In that case $\rho=1/\lambda>1$, so the Lagrangian penalty is high on gain-seeking or excessive-risk-seeking paths.
	\end{rem}
	
	\subsection{Range Dependence: Why the Gap Attenuates}\label{subsec:Why_gap_Attenuates}
	
	The crucial innovation is that the Lagrangian penalty coefficient $\rho$, and hence the reciprocal index $\lambda=1/\rho$, is range-dependent. For the WTA-WTP application, attenuation of the endowment-effect wedge means
	\begin{equation}
		\rho(m) \uparrow 1
		\qquad\text{or, equivalently,}\qquad
		\lambda(m) \downarrow 1
	\end{equation}
	along the relevant transaction path. Substituting $\rho(m)\to1$, equivalently $\lambda(m)\to1$, into \eqref{eq:WTP_Equation} and \eqref{eq:WTA_Equation} yields the neutral benchmark:
	\begin{equation}
		v_\ell(\text{WTP})\to v_g(m),\qquad
		v_g(\text{WTA})\to v_\ell(m).
	\end{equation}
	Under the common normalization $v_g=v_\ell$ on the transaction scale, both converge to the same value of the object.
	
	Formally, for a large transaction scale $M$ satisfying the assumptions of \Cref{lem:LargeTransaction}, we have:

	\begin{equation}
		\lim_{M\to \infty} \bigl(\text{WTA}(M)-\text{WTP}(M)\bigr) = 0.
		\label{eq:Attenuation}
	\end{equation}
	
	The gap attenuates. Small-stakes endowment effects need not imply equally large WTA-WTP gaps at larger or more familiar transaction scales.
	
	\begin{rem}[The Rabin Calibration Paradox]
		The WTA-WTP attenuation result is related to, but distinct from, the Rabin calibration implication \citep{Rabin2000}. In the Rabin application, the large-stakes result concerns favorable risky gambles. In the WTA-WTP application, the relevant neutral benchmark is instead $\lambda=\rho=1$. The common point is that WRDU does not impose a fixed Lagrangian penalty coefficient across scales.
	\end{rem}
	
	\subsection{Summary of the Structural Mechanism}
	
	\begin{table}[htbp]
		\centering
		\caption{Structural Mechanism: WTA-WTP Gap and Its Attenuation}
		\label{tab:Mechanism}
		\small
		\begin{tabular}{lcccc}
			\toprule
			Transaction Scale & Range & $\lambda = v_\ell'/v_g'$ & $\rho=1/\lambda$ & Structural Implication \\
			\midrule
			Small & Narrow & $\lambda > 1$ & $\rho<1$ & WTA $>$ WTP \\
			Attenuated & Wider/familiar & $\lambda \to 1$ & $\rho\to1$ & WTA $\approx$ WTP \\
			\bottomrule
		\end{tabular}
	\end{table}
	
	\section{Illustrative Example}\label{sec:Illustrative_example}
	
	This section provides a concrete example to illustrate the structural mechanism. The example is purely illustrative and does not represent a calibration.
	
	\subsection{Specification}
	
	Let the gain subutility be:
	
	\begin{equation}
	v_g(x) = \sqrt{x}, \qquad x \geq 0.
	\end{equation}
	
	Let the loss subutility (defined on loss magnitudes) be:
	
	\begin{equation}
	v_\ell(x) = \sqrt{x}, \qquad x \geq 0.
	\end{equation}
	
	Both are strictly increasing, strictly concave, and satisfy $v_g(0) = v_\ell(0) = 0$.
	
	The loss-aversion index is:
	
	\begin{equation}
	\lambda(x_r; y, z) = \frac{v_\ell'(y)}{v_g'(z)} = \frac{1/(2\sqrt{y})}{1/(2\sqrt{z})} = \frac{\sqrt{z}}{\sqrt{y}}.
	\end{equation}
	
	This $\lambda$ is range-dependent: it increases with the gain magnitude $z$ and decreases with the loss magnitude $y$.
	
	\subsection{Small Stake}
	
	Consider buying or selling one unit ($m = 1$). Let the relevant range be small, with $y = z = 1$. Then:
	
	\begin{equation}
	\lambda = \frac{\sqrt{1}}{\sqrt{1}} = 1.
	\end{equation}
	
	\paragraph{WTP:}
	\begin{align}
		v_\ell(\text{WTP}) &= \frac{1}{\lambda} \, v_g(1) = 1 \cdot 1 = 1\\
		\sqrt{\text{WTP}} &= 1 \quad \Longrightarrow \quad \text{WTP} = 1.
	\end{align}
	
	\paragraph{WTA:}
	\begin{align}
		v_g(\text{WTA}) &= \lambda v_\ell(1) = 1 \cdot 1 = 1\\
		\sqrt{\text{WTA}} &= 1 \quad \Longrightarrow \quad \text{WTA} = 1.
	\end{align}
	
	Thus WTA $=$ WTP. This is the case where the gap is zero because $\lambda = 1$.
	
	To illustrate an endowment-effect wedge, choose asymmetric subutilities: $v_g(x) = \sqrt{x}$ as before, but $v_\ell(x) = x^{0.8}$. This makes the loss-side marginal utility higher than the gain-side marginal utility at the normalized benchmark. Then:
	
	\begin{equation}
		\lambda(1;1,1) = \frac{v_\ell'(1)}{v_g'(1)} = \frac{0.8}{0.5} = 1.6 > 1.
	\end{equation}
	
	\paragraph{WTP:}
	\begin{align}
	v_\ell(\text{WTP}) &= \frac{1}{\lambda} \, v_g(1) = \frac{1}{1.6} \cdot 1 = 0.625\\
	(\text{WTP})^{0.8} &= 0.625 \quad \Longrightarrow \quad \text{WTP} = (0.625)^{1.25} \approx 0.555.
	\end{align}
	
	\paragraph{WTA:}
	\begin{align}
		v_g(\text{WTA}) &= \lambda v_\ell(1) = 1.6 \cdot 1 = 1.6\\
		\sqrt{\text{WTA}} &= 1.6 \quad \Longrightarrow \quad \text{WTA} = 2.56.
	\end{align}
	
	Thus $\mathrm{WTA}=2.56 > 0.555\approx\mathrm{WTP}$. The illustration generates a WTA-WTP wedge.
	
	\subsection{Large Stake}
	
	The preceding small-stakes calculation used fixed power subutilities and therefore does not describe attenuation. To illustrate attenuation in the WTA-WTP application, keep the same normalized transaction utility for the object and let both the loss-aversion index and the Lagrangian penalty coefficient move toward neutrality:
	
	\begin{equation}
		\lambda(m)=1+\frac{0.2}{(1+m)^2},\qquad
		\rho(m)=\frac{1}{\lambda(m)}.
	\end{equation}
	
	Then $\lambda(1)=1.05$ and $\rho(1)\approx0.9524$, so the small transaction displays a WTA-WTP wedge. For a larger transaction, say $m=100$,
	
	\begin{equation}
		\lambda(100)=1+\frac{0.2}{101^2}\approx1.00002,\qquad
		\rho(100)\approx0.99998.
	\end{equation}
	
	Both indices are now close to one, so the WTA-WTP wedge is correspondingly small. This example is not a CRRA calibration; it only illustrates the transaction-path condition $\lambda(m)\to1$ with $\rho(m)=1/\lambda(m)\to1$.
	
	\subsection{Interpretation of the Example}
	
	The example illustrates the structural mechanism:
	
	\begin{enumerate}
		\item For small stakes, the admissible functional form yields $\lambda > 1$, so WTA $>$ WTP (endowment effect).
		\item For larger or more familiar transactions, $\lambda(m)$ moves toward one and $\rho(m)=1/\lambda(m)$ also moves toward one, so the WTA-WTP gap attenuates.
	\end{enumerate}
	
	The example is purely illustrative; the structural result does not depend on any specific parametric form.
	
	The appendix contains four plots that visualize the same example. \Cref{fig:WTA_WTP_lambda_surface} plots the range-dependent index $\lambda=v_\ell'/v_g'$ for the neutral and endowment-effect cases. \Cref{fig:WTA_WTP_rho_surface} plots the Lagrangian penalty coefficient $\rho=1/\lambda$, showing where gain-seeking or loss attenuation is penalized. \Cref{fig:WTA_WTP_attenuation_path} tracks the transaction path for $\lambda(m)$, $\rho(m)$, WTP, WTA, and the WTA-WTP gap. \Cref{fig:WTA_WTP_summary} summarizes the mechanism by showing how both indices return toward neutrality as the wedge attenuates.
	
	\section{Conclusion}\label{sec:Conclusion}
	
	This paper has axiomatized a reference-dependent WRDU representation for studying the WTA-WTP gap under objective probabilities. The representation uses the Lagrangian penalty coefficient $\rho$ on the loss-side component, and the reciprocal index $\lambda=1/\rho$ reports the standard WTA-WTP loss-aversion comparison. The gap arises from the structural asymmetry of the indifference conditions for buying and selling. When $\lambda>1$, the structural equations generate WTA $>$ WTP on the normalized admissible transaction class. When $\lambda<1$, $\rho>1$ is high and identifies gain-seeking or excessive-risk-seeking paths. Because $\rho$ and its reciprocal index $\lambda$ are range-dependent rather than fixed, the WTA-WTP gap attenuates along transaction paths for which both return toward their neutral value one.
	
	The structural explanation does not depend on CRRA utility or probability weighting. It relies on the WRDU gain-loss representation, range dependence of the Lagrangian coefficient, and a normalization that makes WTA and WTP comparable. The result is complementary to other explanations such as Cautious Utility \citep{CerreiaVioglioDillenbergerOrtoleva2024}, which attributes the gap to uncertainty about trade-offs and caution. The range-dependent attenuation prediction also gives the model a direct empirical implication for transaction environments in which reference points become familiar, market-like, or repeated.
	
	The contribution is therefore deliberately restricted: it is a representation and comparative-statics result for a six-axiom WRDU transaction class, not a claim that all endowment effects or all reference-dependent behavior have the same source.
	
	\appendix
	
	\section{Appendix: Proofs of Main Results}\label{app:MainProofs}
	
	\subsection{Proof of the WRDU Representation Theorem}\label{app:Proof_WRDURepresentation}
	
	\begin{proof}[Proof of \Cref{thm:WRDURepresentation}]
		
		We proceed in five steps.
		
		\medskip
		\noindent
		\textbf{Step 1: Region-wise expected-utility representation.}
		
		By Completeness, Transitivity, and Continuity, $\succeq$ is a continuous weak order on the set of finite lotteries over the compact outcome set $X$.
		
		By Weak Independence, restricted to mixtures that preserve the reference partition and do not move outcomes across gain and loss regions, preferences satisfy the von Neumann-Morgenstern independence condition on each region separately.
		
		Fix a reference point $x_r$. On lotteries whose support lies entirely in the gain region $C_g(X)\cup\{x_r\}$, Weak Independence implies the existence of a continuous affine representation
		\begin{equation}
			V_g(L)=\sum_{x\in C_g(X)} p_x u_g(x),
		\end{equation}
		for some strictly increasing function $u_g$.
		
		Similarly, on lotteries supported in the loss region $C_\ell(X)\cup\{x_r\}$, there exists
		\begin{equation}
			V_\ell(L)=\sum_{x\in C_\ell(X)} p_x u_\ell(x),
		\end{equation}
		for some strictly increasing function $u_\ell$.
		
		Strict monotonicity follows from monotonicity of preferences over degenerate lotteries.
		
		\medskip
		\noindent
		\textbf{Step 2: Reference anchoring and gain-loss normalization.}
		
		By the Reference Partition axiom, $x_r$ serves as the anchor dividing gains and losses. Because outcomes equal to the reference point produce neither gain nor loss, we normalize
		\begin{equation}
			u_g(x_r)=0,
			\qquad
			u_\ell(x_r)=0.
		\end{equation}
		
		Define gain and loss subutilities on magnitudes:
		\begin{align}
			v_g(z)&=u_g(x_r+z)-u_g(x_r), \quad z\ge0,\\
			v_\ell(y)&=u_\ell(x_r)-u_\ell(x_r-y), \quad y\ge0.
		\end{align}
		
		Then $v_g(0)=v_\ell(0)=0$. Strict concavity is imposed as part of the admissible WRDU transaction class; it is not implied by continuity alone.
		
		\medskip
		\noindent
		\textbf{Step 3: Penalized aggregation across regions.}
		
		Consider a lottery $L$ with support spanning both gain and loss regions.
		
		The gain-loss lottery-pair part of Weak Independence identifies how the separately represented gain and loss components are recombined once the reference partition is fixed. It requires the loss-side component to enter the aggregate comparison through a positive range-dependent coefficient.
		
		Range Dependence further requires that this coefficient be a function of the relevant gain and loss magnitudes.
		
		Thus preferences over such lotteries admit a representation of the form
		\begin{equation}
			V(L)=
			\sum_{x\in C_g(X)} p_x v_g(x-x_r)
			-\rho(x_r;y,z)
			\sum_{x\in C_\ell(X)} p_x v_\ell(x_r-x),
		\end{equation}
		for some positive range-dependent Lagrangian penalty coefficient $\rho(x_r;y,z)$. Define the reciprocal loss-aversion index by
		\begin{equation}
			\lambda(x_r;y,z)=\frac{1}{\rho(x_r;y,z)}.
		\end{equation}
		The coefficient $\rho$ is the Lagrangian shadow-price penalty on the loss-side component, while $\lambda=1/\rho$ is the reciprocal index used below to report WTA-WTP comparative statics. This yields the representation
		\begin{equation}
			V^\rho(L)=
			\sum_{x\in C_g(X)} p_x v_g(x-x_r)
			-\rho(x_r;y,z)
			\sum_{x\in C_\ell(X)} p_x v_\ell(x_r-x).
		\end{equation}
		
		\medskip
		\noindent
		\textbf{Step 4: Characterization of the endogenous loss-aversion index.}
		
		By the Reference Partition axiom, the reference point $x_r$ is selected as a maximizer of the Lagrangian penalty problem
		\begin{equation}
			r \mapsto
			\mathcal L(r;\rho)
			=
			v_g(z(r))-\rho(r)v_\ell(y(r)),
			\qquad
			\rho(r)=\frac{1}{\lambda(r;y(r),z(r))}.
		\end{equation}
		
		At an interior stationary root, the local loss-aversion index is measured by the marginal loss-gain ratio:
		\begin{equation}
			\lambda(x_r;y,z)=\frac{v_\ell'(y)}{v_g'(z)}.
		\end{equation}
		
		Thus the loss-aversion index is determined by the ratio of marginal utilities across domains. Values above one correspond to loss aversion in the standard sense; values below one correspond to gain seeking or loss attenuation. If the first-order condition has several roots, the reference-point axiom selects the root that maximizes $\mathcal L(r;\rho)$, so the model uses optimality rather than stationarity alone.
		
		\medskip
		\noindent
		\textbf{Step 5: Semi-affine uniqueness.}
		
		Suppose $(\tilde v_g,\tilde v_\ell,\tilde\rho,\tilde\lambda)$ is another representation satisfying
		\begin{equation}
			\tilde V^\rho(L)=
			\sum_{x\in C_g(X)} p_x \tilde v_g(x-x_r)
			-\tilde\rho(x_r;y,z)
			\sum_{x\in C_\ell(X)} p_x \tilde v_\ell(x_r-x).
		\end{equation}
		
		Because both representations agree on lotteries supported entirely in $C_g(X)$, the standard von Neumann-Morgenstern uniqueness theorem implies
		\begin{equation}
			\tilde v_g = a v_g + b
		\end{equation}
		for some constants $a>0$, $b\in\mathbb R$.
		
		Anchoring at the reference point imposes $\tilde v_g(0)=v_g(0)=0$, hence $b=0$ and
		\begin{equation}
			\tilde v_g = a v_g.
		\end{equation}
		
		An identical argument applies in the loss region, yielding
		\begin{equation}
			\tilde v_\ell = c v_\ell.
		\end{equation}
		for some $c>0$. The common-scale normalization used for WTA-WTP comparisons sets $c=a$.
		
		Substituting into the full representation then shows that the penalty index must satisfy
		\begin{equation}
			\tilde\rho(x_r;y,z)=\rho(x_r;y,z),
			\qquad
			\tilde\lambda(x_r;y,z)=\lambda(x_r;y,z).
		\end{equation}
		
		Therefore the normalized representation is unique up to the common positive scaling factor $a>0$.
		Independent additive transformations of $v_g$ and $v_\ell$ are not admissible because they would violate the anchoring condition at $x_r$ and alter the gain-loss partition.
		
		\medskip
		
		This establishes existence and semi-affine uniqueness of the WRDU representation.
	\end{proof}
	
	\subsection{Proof of the Fixed-Index Proposition}\label{app:Proof_FixedLambdaImpossibility}
	
	\begin{proof}[Proof of \Cref{prop:FixedLambdaImpossibility}]
		For each $m>0$, the indifference equations are
		\begin{equation}
			v_\ell(\mathrm{WTP}(m))=\frac{1}{\bar\lambda}\, v_g(m)=\bar\rho v_g(m),
			\qquad
			v_g(\mathrm{WTA}(m))=\bar\lambda v_\ell(m)=\frac{1}{\bar\rho}v_\ell(m),
			\qquad \bar\rho=\frac{1}{\bar\lambda}.
		\end{equation}
		Strict monotonicity implies
		\begin{equation}
			\mathrm{WTP}(m)=v_\ell^{-1}\!\left(\frac{1}{\bar\lambda}\, v_g(m)\right),
			\qquad
			\mathrm{WTA}(m)=v_g^{-1}\!\left(\bar\lambda v_\ell(m)\right).
		\end{equation}
		If $\bar\lambda\neq1$, equivalently $\bar\rho\neq1$, the two equations differ by a fixed multiplicative distortion. Since that distortion is independent of $m$, the fixed reciprocal index, or equivalently the fixed Lagrangian coefficient $\bar\rho$, provides no mechanism by which the wedge moves toward the neutral benchmark as transaction scale changes. Thus a fixed reciprocal loss-aversion index cannot generate structural attenuation by itself.
	\end{proof}
	
	\section{Appendix: Proof of the Structural Mechanism}
	
	This appendix provides a formal statement of the WTA-WTP mechanism. The result is stated conditionally because WTA and WTP compare different dimensions and therefore require a common normalization.
	
	\subsection{Setup}
	
	Let $v_g$ and $v_\ell$ be strictly increasing, strictly concave $C^2$ functions with $v_g(0)=v_\ell(0)=0$. Let
	\begin{equation}
		\lambda(x_r;y,z)=\frac{v_\ell'(y)}{v_g'(z)}
	\end{equation}
	denote the WRDU virtual loss-aversion index, and let $\rho(x_r;y,z)=1/\lambda(x_r;y,z)$ denote the corresponding Lagrangian penalty coefficient.
	
	\subsection{Lemma: Endowment Effect for Small Transactions}
	
	\begin{lem}\label{lem:SmallTransaction}
		Suppose the normalized transaction path satisfies $\lambda(m)>1$ for sufficiently small $m>0$ and the inverse-utility comparisons induced by \eqref{eq:WTP_Equation} and \eqref{eq:WTA_Equation} are ordered on the common transaction scale. Then
		\begin{equation}
			\text{WTA}(m)>\text{WTP}(m).
		\end{equation}
	\end{lem}
	
	\begin{proof}
		Since $\lambda(m)>1$, Equation \eqref{eq:WTP_Equation} lowers the WTP side relative to the neutral benchmark, while \eqref{eq:WTA_Equation} raises the WTA side. The maintained ordering of the inverse-utility comparisons gives WTA$(m)>$WTP$(m)$.
	\end{proof}
	
	\subsection{Lemma: Attenuation Along Neutralizing Transaction Paths}
	
	\begin{lem}\label{lem:LargeTransaction}
		Suppose the normalized transaction path satisfies $\lambda(M)\to1$ as $M$ grows and that the neutral benchmark has WTA$(M)=$WTP$(M)$. Then
		\begin{equation}
			\lim_{M\to\infty}\bigl(\text{WTA}(M)-\text{WTP}(M)\bigr)=0.
		\end{equation}
	\end{lem}
	
	\begin{proof}
		By \eqref{eq:WTP_Equation} and \eqref{eq:WTA_Equation}, WTP and WTA are continuous functions of $\lambda$ whenever $v_g$ and $v_\ell$ are strictly increasing and continuous. As $\lambda(M)\to1$, both indifference equations converge to their neutral forms. Since the neutral benchmark has WTA$(M)=$WTP$(M)$, the gap vanishes.
	\end{proof}
	
	\subsection{Theorem}
	
	\begin{thm}
		Under the WRDU transaction assumptions above, there exists an admissible class of transaction paths such that:
		\begin{enumerate}
			\item For sufficiently small transactions, WTA $>$ WTP;
			\item Along neutralizing transaction paths, the gap attenuates: WTA $-$ WTP $\to0$.
		\end{enumerate}
	\end{thm}
	
	\begin{proof}
		The proof follows directly from the two preceding lemmas.
	\end{proof}

	\section{Appendix: Illustrative WTA-WTP Plots}

	The following figures reproduce the numerical illustrations from Section 5. They are not calibrations; they are graphical vignettes showing how the WRDU loss-aversion index, the Lagrangian penalty coefficient, and the WTA-WTP wedge move along the illustrative transaction path.

	\begin{figure}[!htpb!]
		\centering
		\captionof{figure}{Range-dependent WRDU loss-aversion index in the illustrative WTA-WTP example}
		\label{fig:WTA_WTP_lambda_surface}
		\includegraphics[page=1,width=0.92\textwidth]{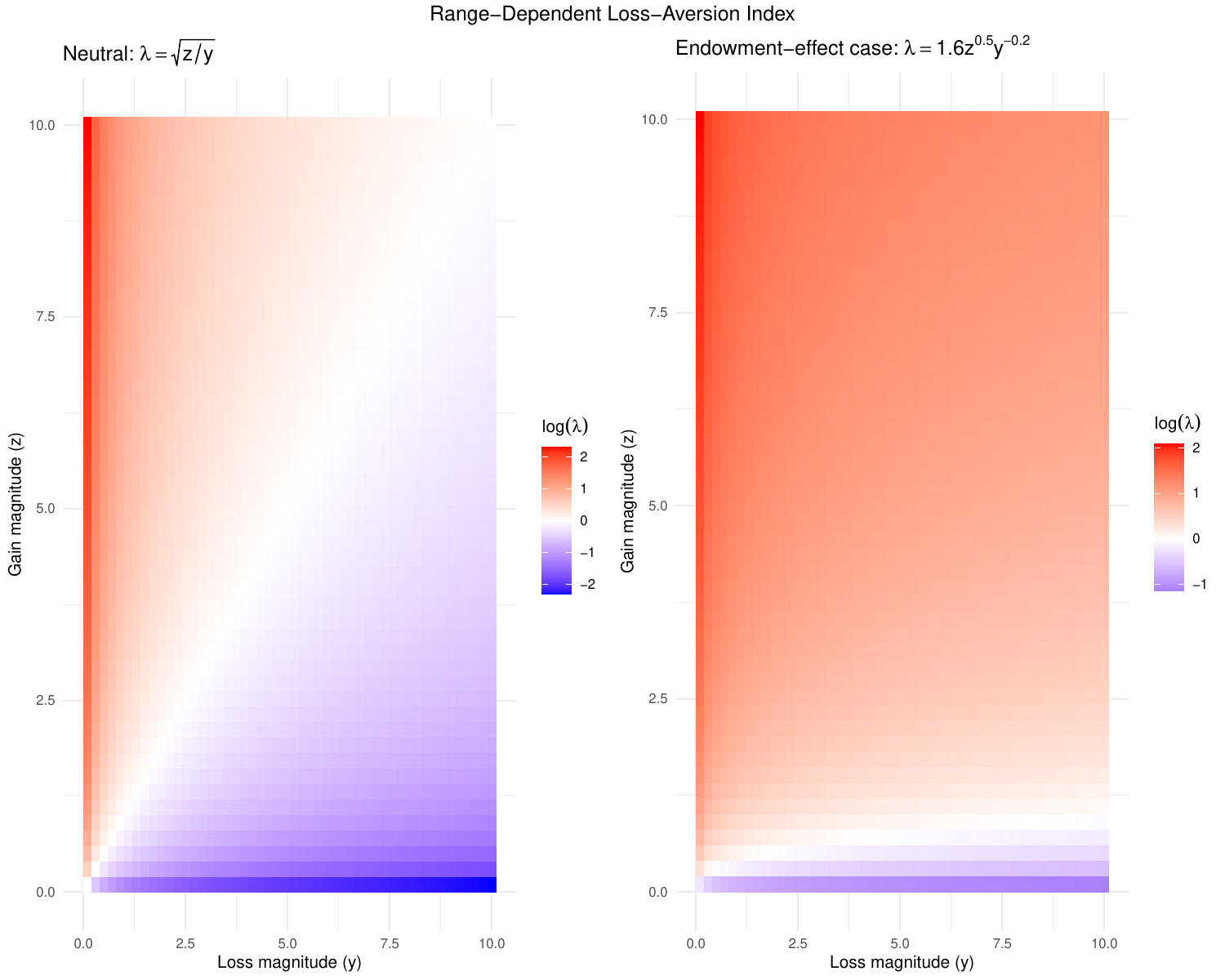}
		\begin{minipage}{0.85\textwidth}
			The two panels plot $\log(\lambda)$ over loss and gain magnitudes. The neutral case uses identical square-root gain and loss subutilities, while the endowment-effect case uses $v_g(x)=\sqrt{x}$ and $v_\ell(x)=x^{0.8}$. The figure shows that $\lambda=v_\ell'(y)/v_g'(z)$ is range-dependent rather than a fixed primitive.
		\end{minipage}
	\end{figure}

	\begin{figure}[!htpb!]
		\centering
		\captionof{figure}{Lagrangian WRDU penalty coefficient in the illustrative WTA-WTP example}
		\label{fig:WTA_WTP_rho_surface}
		\includegraphics[page=2,width=0.92\textwidth]{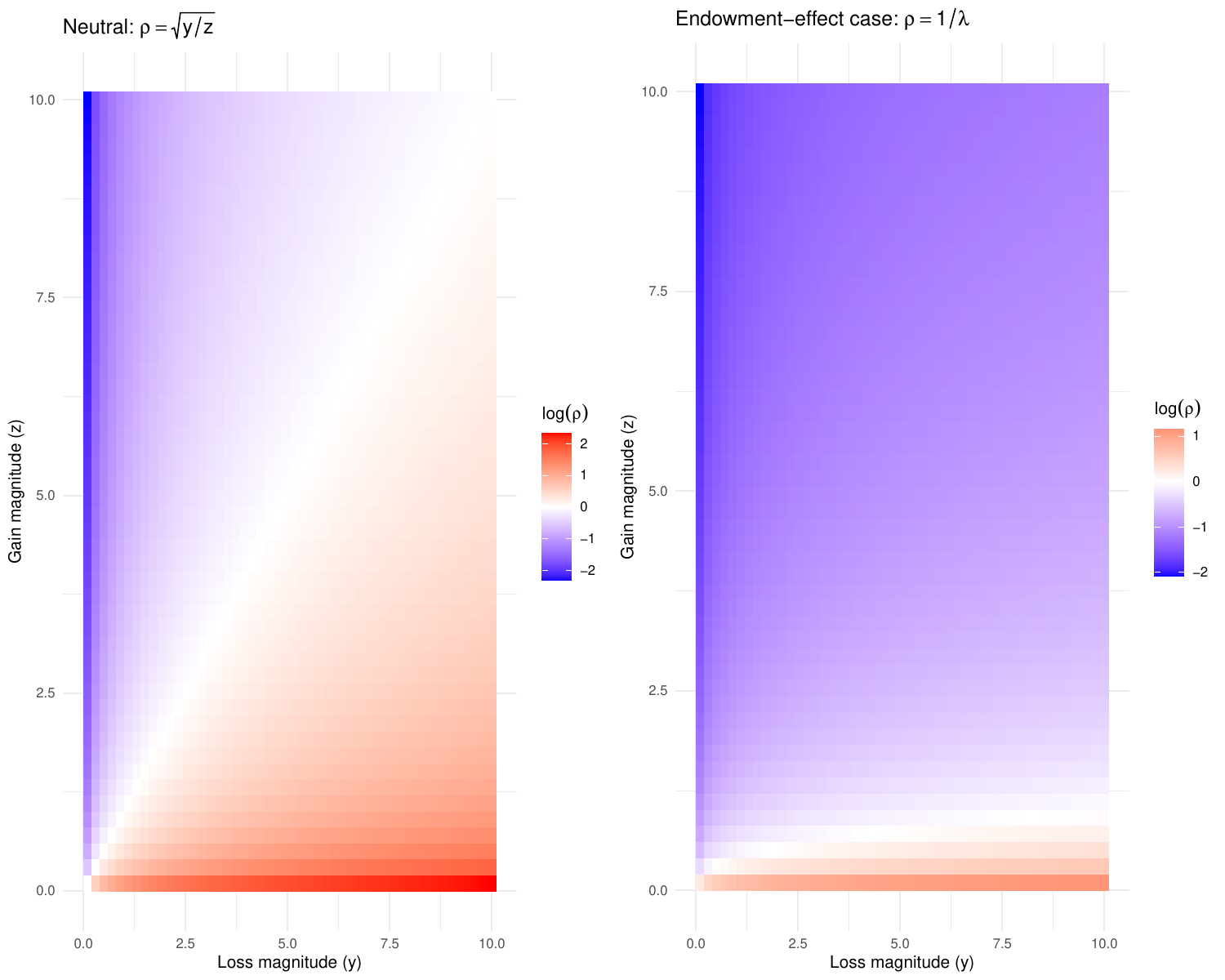}
		\begin{minipage}{0.85\textwidth}
			The panels plot the Lagrangian penalty coefficient $\rho=1/\lambda$. Under the standard WTA-WTP loss-aversion convention, $\lambda>1$ is loss aversion, while $\lambda<1$ is gain seeking or loss attenuation. Therefore $\rho>1$ identifies the gain-seeking region and penalizes excessive-risk-seeking paths; in the WTA-WTP loss-aversion region, $\rho<1$ and moves toward one as the wedge attenuates.
		\end{minipage}
	\end{figure}

	\begin{figure}[!htpb!]
		\centering
		\captionof{figure}{Attenuation path for $\lambda(m)$, $\rho(m)$, WTP, WTA, and the WTA-WTP gap}
		\label{fig:WTA_WTP_attenuation_path}
		\includegraphics[page=3,width=0.92\textwidth]{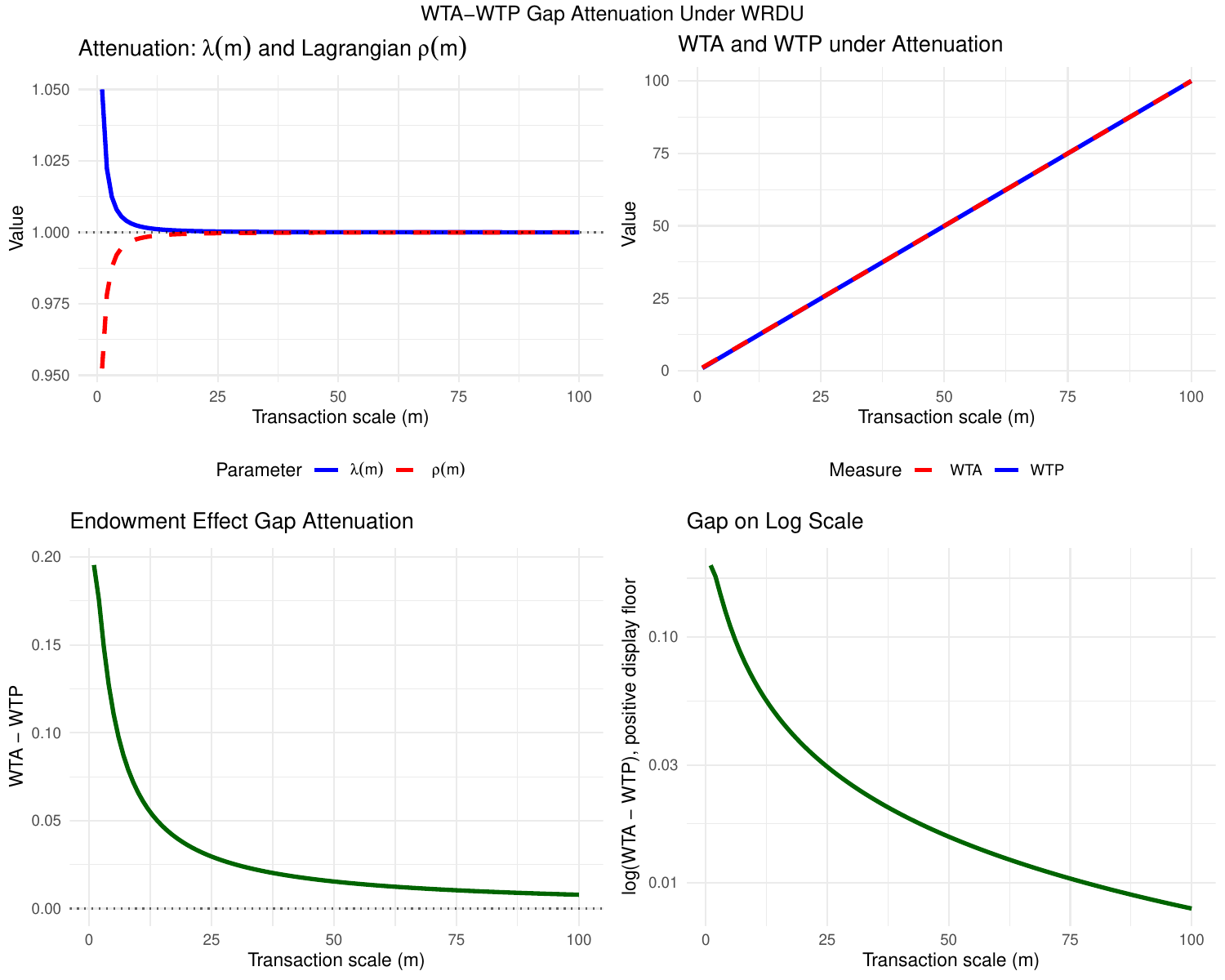}
		\begin{minipage}{0.85\textwidth}
			The four panels report the illustrative transaction path used to visualize attenuation. Along the path, $\lambda(m)$ moves toward one and the Lagrangian penalty coefficient $\rho(m)=1/\lambda(m)$ also moves toward its neutral value. Under the common normalization used in the illustration, WTA and WTP converge and the WTA-WTP wedge becomes small.
		\end{minipage}
	\end{figure}

	\begin{figure}[!htpb!]
		\centering
		\captionof{figure}{Summary of WRDU attenuation in the illustrative WTA-WTP example}
		\label{fig:WTA_WTP_summary}
		\includegraphics[page=4,width=0.88\textwidth]{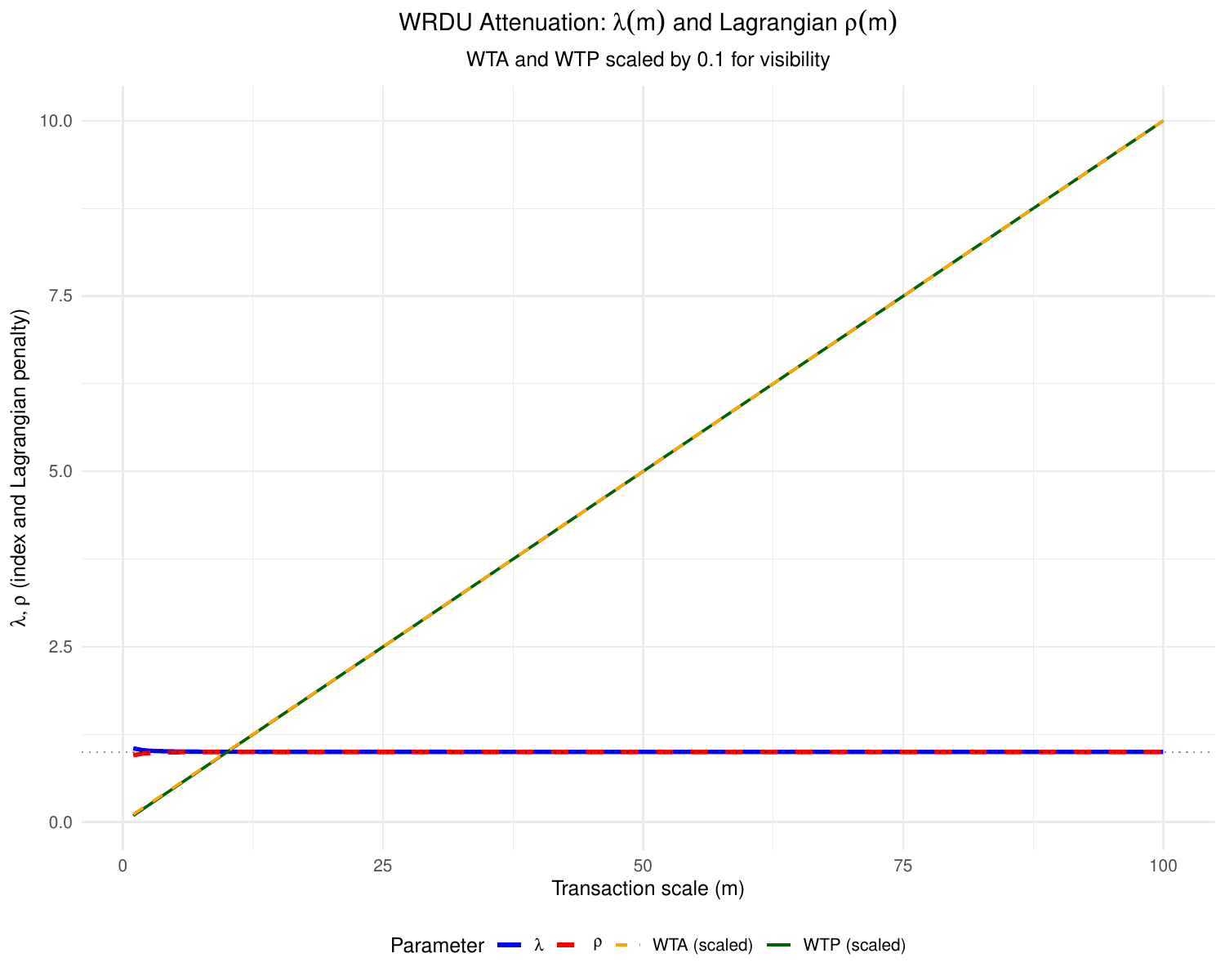}
		\begin{minipage}{0.85\textwidth}
			The summary plot overlays the path of $\lambda$, the Lagrangian penalty coefficient $\rho$, and scaled WTA and WTP values. The figure is intended only as a compact visualization of the mechanism: $\lambda$ is high enough to generate a small-scale WTA-WTP wedge, while $\rho=1/\lambda$ records the Lagrangian penalty and both indices return toward neutrality as the transaction path widens.
		\end{minipage}
	\end{figure}

	\clearpage

\singlespace
\bibliographystyle{chicago}
\addcontentsline{toc}{section}{References}
\bibliography{../LossAversionGeneral}

\end{document}